\documentclass[11pt]{article}

\usepackage{booktabs} 

\usepackage{amsmath}
\usepackage{amsthm}
\usepackage{amsfonts}
\usepackage{graphicx}
\usepackage{url}
\usepackage{color}
\usepackage{bbm}
\usepackage{fullpage}
\usepackage{bm}
\usepackage{authblk}
\usepackage[square, comma, numbers,sort&compress]{natbib}

\newcommand{\ubar}[1]{\text{\b{$#1$}}}
\newcommand{\E}{\mathbb{E}}
\newcommand{\reals}{\mathbb{R}}
\newcommand{\Var}{\mathrm{Var}}

\newtheorem{definition}{Definition}
\newtheorem{theorem}{Theorem}
\newtheorem{claim}{Claim}
\newtheorem{lemma}{Lemma}
\newtheorem{corollary}{Corollary}

\title{Downstream Effects of Affirmative Action}

\author[1]{Sampath Kannan\thanks{Supported in part by NSF grant AF-1763307 and a grant from the Quattrone Center for the Fair Administration of Justice}}
\author[1]{Aaron Roth\thanks{Supported in part by NSF grants CNS-1253345, AF-1763307, and a grant from the Quattrone Center for the Fair Administration of Justice}}
\author[2]{Juba Ziani\thanks{Supported in part by NSF grants CNS-1331343 and CNS-1518941, the US-Israel Binational Science Foundation grant 2012348, and the Linde Graduate
Fellowship at Caltech.}}
\affil[1]{University of Pennsylvania}
\affil[2]{California Institute of Technology}

\begin{document}
\maketitle

\begin{abstract}
We study a two-stage model, in which students are 1) admitted to college on the basis of an entrance exam which is a noisy signal about their qualifications (type), and then 2) those students who were admitted to college can be hired by an employer as a function of their college grades, which are an independently drawn noisy signal of their type. Students are drawn from one of two populations, which might have different type distributions. We assume that the employer at the end of the pipeline is rational, in the sense that it computes a posterior distribution on student type conditional on all information that it has available (college admissions, grades, and group membership), and makes a decision based on posterior expectation. We then study what kinds of fairness goals can be achieved by the college by setting its admissions rule and grading policy. For example, the college might have the goal of guaranteeing equal opportunity across populations: that the probability of passing through the pipeline and being hired by the employer should be independent of group membership, conditioned on type. Alternately, the college might have the goal of incentivizing the employer to have a group blind hiring rule. We show that both goals can be achieved when the college does not report grades. On the other hand, we show that under reasonable conditions, these goals are impossible to achieve even in isolation when the college uses an (even minimally) informative grading policy.
\end{abstract}

\section{Introduction}

For a variety of reasons, including unequal access to primary education, family support, and enrichment activities, different demographic groups can vary widely in their level of preparation by the time they reach their senior year of high school, when they apply for college. In an attempt to correct for this unfortunate reality, many colleges in the United States follow some sort of affirmative action policy in their admissions, which is to say, their admissions decisions explicitly take demographics into account. What is often unstated (and perhaps not even explicitly considered by the colleges) is what exactly the long term goals of these policies are, beyond the short term goal of having a diverse freshman class. In this paper, we consider two explicit goals, and study the extent to which they can be met in a simple two stage model:
\begin{enumerate}
\item \textbf{Equal opportunity}: The probability that an individual is accepted to college \emph{and then} ultimately hired by an employer may depend on an individual's type, but conditioned on their type, should not depend on their demographic group.
\item \textbf{Elimination of Downstream Bias}: Rational employers selecting employees from the college population should not make hiring decisions based on group membership.
\end{enumerate}

Neither of these desiderata will necessarily be achieved by admissions rules that ignore demographic information. For example, suppose college admissions is set by a uniform admissions threshold on entrance exam scores. Assuming these scores are equally informative about all groups, this will guarantee that conditioned on a student's type, whether or not she is \emph{admitted to college} will be independent of her group membership, but it does not imply that whether or not she is ultimately hired is independent of her group! This is because exam scores are only a noisy signal about student type. Therefore, if two groups have different prior distributions on type, they will have different posterior distributions on type when conditioned on being admitted to college according to a group-blind admissions rule. The result will be that a Bayesian employer will insist that students from a group with lower mean or higher variance will have to cross a higher threshold on their college grades in order to be hired. In addition to incentivizing explicit group-based discrimination by the employer, this also results in a failure of equal opportunity for the students, because once admitted to college, two individuals of the same type might have to cross different grade thresholds in order to be hired. Thus, a simple ``group blind'' admissions rule fails to achieve either goal 1 or 2 as laid out above. In this paper, we study the extent to which these goals can be achieved via other means available to the college: in particular, how it admits and grades students.

\subsection{Limitations of our Model}
When interpreting our results, it is important to understand the scope and limitations of our model. First, this paper considers fairness goals that are limited to preventing inequity from being further propagated --- treating opportunities at the high school level and earlier as fixed --- and that do not attempt to correct for past inequity. This manifests itself in that our ``equal opportunity'' goal takes as given that the prospects for employment may ``fairly'' vary as a function of an individual's type \emph{at the time at which they apply for college}, and does not attempt to address or correct the historical forces that might have resulted in different groups having different type distributions to begin with. Attempting to correct for this kind of historical inequity would require a ``value-added model'' of education, in which colleges can \emph{change} the type distributions of their student population either through the direct effect of education, or through a second order effect on student behavior before they apply. In our model, colleges do not change student types, they only serve as signaling mechanisms. Similarly, our ``equal opportunity'' goal aims to equalize the probability that students are hired conditioned on their types --- but one might reasonably instead ask for a corrective notion of fairness, in which the probability of passing through the pipeline is \emph{higher} for the historically disadvantaged group conditioned on type. We do not consider this.

Our model also ignores the possibility that exam scores and grades are themselves \emph{biased}. We explicitly assume the opposite --- that exam scores and grades are unbiased estimators of student types, for both groups. If instead exam scores were systematically biased downwards for one group, then the response of a rational employer to an admissions policy would be very different --- because students who made it through the college pipeline \emph{despite} negative bias would have a higher relative posterior probability of having a high type. There is evidence that effects of this sort are real \cite{bohren2017dynamics}.

The two kinds of fairness goals that we study do not speak to the size of the student of employee population coming from each group. For example, in principle, one could satisfy both the equal opportunity and elimination-of-downstream-bias goals that we propose, but at a cost of employing very few individuals from one of the groups. However, we show that even without an additional goal of having large representation from both groups, the fairness goals we set out cannot generally be achieved.

Finally, we assume that employers are single-minded expectation maximizers, with no explicit desire for fairness or diversity. Of course this is often not the case.

Despite these limitations and simplifying assumptions, we find that in the model we study, many natural fairness goals are already impossible. We think that these negative results are likely to persist in more complex models that attempt to capture additional realism.

\subsection{Our Model and Results}
We consider a simple model of admissions, grading and hiring that views the role of colleges only as a means of signaling quality and performing a gatekeeping function, rather than as providing explicit value added\footnote{This is consistent with the signaling view of the role of colleges in the economics literature, beginning with \cite{spence}}. We consider two \emph{groups} representing pre-defined subsets of the population, divided according to socio-economic or other demographic lines. Each student from group $i$ is endowed with a \emph{type} $t$, which is drawn independently from a Gaussian type distribution $P_i$ that is dependent on the students' group membership. A student's type ultimately measures her value to an employer. We model employers as having a fixed  cost $C$ for hiring an individual, and a gain that is proportional to their type. If the employer hires an individual who has type $t$, they obtain utility $t-C$. A college can choose an admissions rule and a grading policy. Although students types are unobservable, each student has an admissions exam score that is an observable unbiased estimator of their type. We model exam scores as being distributed as a unit variance Gaussian, centered at the student's type.  An admissions policy for the school is a mapping between exam scores and admissions probabilities. We allow schools to set different admissions policies for different groups, but for most of our results, we require the natural condition that admissions probabilities within a group be monotonically non-decreasing in exam scores\footnote{A non-monotone admissions rule would have the property that sometimes a student with a lower exam score would have a higher probability of admission that a student with a higher exam score. Non-monotonicity within a group is highly undesirable, because it would give some students a perverse incentive to intentionally try and lower their exam scores. If such incentives were present, it would no longer be reasonable to model exam scores as unbiased estimators of student types.}. Deterministic monotone admissions policies simply correspond to setting admissions thresholds based on exam scores. For simplicity, in the body of the paper, we restrict attention to deterministic admissions rules, but in the Appendix, we extend our results to cover probabilistic admissions rules as well.

Schools may also set a grading policy. A grade is also modeled as a Gaussian centered at a student's true type, but the school may choose the variance of the distribution. We assume that a student's grade is conditionally independent of her entrance exam score, conditioned on her type. One limiting extreme (infinite variance) corresponds to committing not to report grades at all. This limiting case is actually achievable because schools can simply opt not to share grades ---  in fact, this practice has been adopted at several top business schools \cite{nogrades}.  At the other limiting extreme, types are perfectly observable. This extreme is generally not achievable, and we do not consider it in this paper. In between, the school can modulate the strength of the signal that employers get about student type, beyond the simple indicator that they were admitted to college.

Employers know the prior distributions $P_i$ on student types, as well as the admissions and grading policy of the school. They are rational expectation maximizers. When deciding whether or not to hire a student, they will condition on all information available to them --- a student'a group membership, the fact that she was admitted to college under the college's admissions policy, and the grade that she received under the college's grading policy --- to form a posterior distribution about the student's type. They will hire exactly those students for whom they have positive expected utility under this posterior distribution.

In order to incentivize a particular employer to use a hiring rule that is independent of group membership, it is necessary to set admissions and grading policies such that for every student admitted to the school, and for every grade $g$ that she may receive, the indicator that the conditional expectation of her type $t$ is above the employer's hiring cost $C$ is independent of the student's group membership. If there is uncertainty about what the employer's hiring cost $C$ is, or if there are multiple employers, then it is necessary to guarantee this property for an interval of hiring costs $C \in [C^-, C^+]$ rather than for just a fixed cost. We distinguish these two cases. We call this property Irrelevance of Group Membership (IGM), in the single threshold and multiple threshold case respectively. A seemingly stronger property that we might desire is that the posterior distribution on student types conditional on admission to college is \emph{identical} for both groups. We call this property \emph{strong} Irrelevance of Group Membership (sIGM). Because it symmetrizes the two groups, it in particular guarantees that members of both groups will be treated identically by rational decision makers at any further stage down the decision making pipeline. We show that in the presence of finite, nonzero variance in both exam scores and grades, IGM in the multiple threshold case implies sIGM. Finally, we say that an admissions rule and grading policy satisfy the \emph{equal opportunity} condition, if a student's probability of making it all the way through the pipeline --- i.e. being admitted to college \emph{and then} being hired by the employer, is independent of her group conditioned on her type. Trivially, any group-symmetric admissions policy will satisfy both conditions if the two group type distributions are identical, so for the results that follow, we always assume that the group type distributions are distinct --- differing in their mean, their variance, or both.

First, to emphasize that our impossibility results will crucially depend on the fact that exam scores are only a noisy signal of student ability, we consider the noiseless case, in which college admissions can be decided \emph{directly} as a function of student type (this corresponds to the case in which exam scores have no noise). In this case, we can ``have it all'': there is a simple monotone admissions rule that guarantees both the equal opportunity condition, and satisfies IGM for multiple thresholds --- for any grading policy that the school might choose. After establishing this simple result, in the rest of the paper we move on to the more realistic case in which exam scores are only a noisy signal of student type.

Next, we study what is possible if the college chooses to not report grades at all. In this case, we can also ``have it all'' --- simply by setting a sufficiently high, group independent admissions threshold, a school can achieve both equal opportunity and IGM for multiple thresholds. This gives another view of the effects of practicing grade non-disclosure at highly selective schools \cite{nogrades}.

Finally, in the bulk of the paper, we study the common case in which the college uses informative grades --- i.e. sets the variance of its grade distribution to be some finite value. In this case, we show that it \emph{is} possible to obtain IGM  in the single threshold case, but that no monotone admissions rule can obtain sIGM. Because of the equivalence between sIGM and IGM for the multiple threshold case, this implies that no monotone admissions rule can obtain IGM in the multiple threshold case, even in isolation. Next, we consider the equal opportunity condition. One trivial way to obtain it is to simply admit nobody to college. We show that this is in general the only way in the multiple thresholds case: no non-zero monotone admissions rule can satisfy the equal opportunity condition, even in isolation.

\subsection{Related Work}
Our work fits into two streams of research. Within the recent line of work on algorithmic fairness, the most closely related work is that of Chouldechova \cite{Chou16} and Kleinberg, Mullainathan, and Raghavan \cite{KMR16}. Both of these papers prove the impossibility of simultaneously satisfying certain fairness desiderata in batch classification and regression settings. Broadly speaking, both papers show the impossibility of simultaneously equalizing false positive and false negative rates (related to our equal opportunity goal --- see also \cite{HPS16}) and positive predictive value or calibration (related to our IGM goals). Our work is quite different, however: the goals that we study are not direct properties of the classification rule in question (in our case, the college admissions rule), but instead properties of its downstream effects. And while the work of \cite{Chou16,KMR16} shows the impossibility of simultaneously satisfying these fairness criteria, in our setting, we show that they are often impossible to satisfy even in isolation.

Our paper also fits into an older line of work studying economic models of discrimination and affirmative action, which has its modern roots in \cite{arrow} and \cite{phelps}. For example, Coate and Loury \cite{CL93} and Foster and Vohra \cite{FV92} study two stage models in which students from two different groups (who are a-priori identical) can in the first stage choose whether or not to make a costly investment in themselves, which will increase their value to employers. In the 2nd stage, employers may set a hiring rule that acts on a noisy signal about student quality. These works show the existence of a self-confirming equilibrium, in which only one group makes investments in themselves and are subsequently given employment opportunities, and consider interventions which can escape these discriminatory equilibria. These works can be viewed as studying the ``upstream effects'' of affirmative action policies, and explaining the mechanics by which different student populations may end up with different type distributions. The effect of the interventions proposed in these models is very slow, because it requires a new generation of students to recognize the opportunities made available to them via affirmative action policies and make costly investments in their education in response, well before they enter the job market. In contrast, our work can be viewed as studying the ``downstream effects'' of these policies and examining shorter term effects which can be realized in a time frame that need not be long enough for type distributions to change.

More recently, the computer science community has begun studying fairness desiderata in dynamic models. Jabbari et al study the costs (measured as their effect on the rate of learning) of imposing fairness constraints on learners in general Markov decision processes \cite{JJKMR17}. Hu and Chen \cite{HC18} study a dynamic model of the labor market similar to that of \cite{CL93,FV92} in which two populations are symmetric, but can choose to exert costly effort in order to improve their value to an employer. They study a two stage model of a labor market in which interventions in a ``temporary'' labor market can lead to high welfare symmetric equilibrium in the long run.  Liu et al. \cite{delayed} study a two round model of lending in which lending decisions in the first round can change the type distribution of applicants in the 2nd round, according to a known, exogenously specified function. They study how statistical constraints on the lending rule can improve or harm outcomes as compared to a myopic (i.e. ignoring dynamic effects) profit maximizing rule, and find that for two kinds of interventions, both improvement and harm are possible, depending on the details of how lending effects the type distribution. Finally, \cite{incentives} studied the regulator's problem of providing financial incentives for a lender to satisfy fairness constraints in an online classification setting.

\section{Model}
We consider two populations of students, $1$ and $2$. In population $i \in \{1,2\}$, each student has a type drawn from a Gaussian distribution $P_i = \mathcal{N}\left(\mu_i,\sigma_i^2\right)$ with mean $\mu_i$ and variance $\sigma_i^2$. Since our problem is trivial if $P_1 = P_2$, in this paper we assume always that $P_1 \neq P_2$, i.e. the type distributions differ either in their mean, or their variance, or both. We denote by $T_i$ the random variable that represents the type of a student from population $i$. Throughout the paper,  $\phi$ denotes the probability density function and $\Phi$ the cumulative density  function of a standard normal random variable with mean $0$ and variance $1$.

Each student takes a standardized test (SAT, etc.) and obtains a score given by
\[
S_i = T_i + X
\]
where $X$ follows a normal distribution with mean $0$ and variance $1$, that does not depend on the population $i$, i.e., the student's score is a noisy but unbiased estimate of his type.

Additionally, we consider a university that admits students from both populations. The university designs an admission rule $A_i: \reals \to [0,1]$ for each population $i$, such that a student from population $i$ with score $s$ is accepted with probability $A_i(s)$. We also abuse notation and let $A_i$ denote the binary random variable whose value is $1$ if a student is accepted, and $0$ otherwise. This admission rule is required to be monotone non-decreasing; i.e. an increase in exam score cannot lead to a \emph{decrease} in admissions probability. We say that an admissions rule is deterministic if $A_i(s) \in \{0,1\}$. A deterministic monotone admissions rule is characterized by a threshold $\beta_i$ such that a student is accepted if and only if $S_i \geq \beta_i$. We call such rules``thresholding admissions rules''. We focus on thresholding admissions rules in the body of this paper, but extend our results to probabilistic admissions rules to the Appendix. For simplicity of notation, we will often write $x_i(t)=\Pr \left[A_i = 1 |T_i = t \right]$ (Note that $x_i(t) = \Pr \left[S_i \geq \beta_i |T_i = t \right]$ in the deterministic case).

Every student who is admitted to the university receives a grade, given by:
\[
G_i = T_i + Y
\]
where $Y$ follows a normal distribution with mean $0$ and variance $\gamma^2$ that does not depend on the population $i$. $\gamma$ can be set by the university, and represents the strength of the signal provided by a grading policy\footnote{In actuality, of course, students receive many grades, not just one. But note that when one averages two normally distributed random variables, the result is also normally distributed, but with lower variance. Hence, one way to modulate the variance of a grade signal is to modulate the \emph{number} of grades computed. The more assignments and exams that are graded, the lower the variance of the signal. The fewer that are graded, the higher the variance.}. In our model, the University must commit to a single grading policy to use across groups.

Finally, an employer makes a hiring decision for each student that graduates from the university. The employer knows the priors $P_i$, the admission rules $A_1,~A_2$ used by the school, the grading policy $\gamma$, and observes the grades of the students (as well as the fact that they were admitted to the school). The employer's expected utility for accepting a university graduate from population $i$ with grade $g$ is then given by
\[
\mathbb{E} \left[ T_i |G_i = g, A_i = 1\right] - C
\]
where $C$ is the cost for the employer to hire a student.
The employer hires a university graduate from population $i$ with grade $g$ if and only if
\[
\mathbb{E} \left[ T_i  | G_i = g, A_i = 1\right] \geq C
\]

Throughout the paper, we study the feasibility of achieving the following fairness goals:

\begin{definition}[Equal opportunity]
Equal opportunity holds if and only if the probability of a student being hired by the employer conditional on his type is independent of the student's group. I.e. if for all types $t \in \reals$,
\begin{align*}
&\int_{g} \Pr \left[G_1 = g, A_1 = 1 | T_1 = t \right] \mathbbm{1} \{ \E \left[T_1  | G_1 = g, A_1 = 1 \right] \geq C\} dg
\\& = \int_{g} \Pr \left[G_2 = g, A_2 = 1 | T_2 = t\right] \mathbbm{1} \{ \E \left[T_2  | G_2 = g, A_2 = 1 \right] \geq C\} dg
\end{align*}
\end{definition}

\begin{definition}[Irrelevance of Group Membership]
\textit{Irrelevance of Group Membership (IGM)} holds if and only if, conditional on admission by the school and on grade $g$, the employer's decision on whether to hire a student is independent of the student's group. I.e. if for all grades $g \in \reals$,
\begin{align*}
\E \left[T_1  | G_1 = g, A_1 = 1 \right] \geq C \Leftrightarrow \E \left[T_2  | G_2 = g, A_2 = 1 \right] \geq C
\end{align*}
\end{definition}

We further introduce a robust version of IGM, called \textit{strong} Irrelevance of Group Membership, that symmetrizes the two populations and guarantees that members of both populations will be treated identically by rational decision makers at any further stage of the decision making pipeline.

\begin{definition}[strong Irrelevance of Group Membership]
\textit{Strong Irrelevance of Group Membership (sIGM)} holds if and only if, conditional on admission by the school and on grade $g$, the employer's posterior on a student's type is independent of the student's population. I.e., for all $g \in \reals$, for all $t \in \reals$,
\begin{align*}
\Pr \left[T_1 = t  | G_1 = g, A_1 = 1 \right] = \Pr \left[T_2 = t  | G_2 = g, A_2 = 1 \right]
\end{align*}
\end{definition}
We note that sIGM holds if and only if the posterior on students' types conditional on admission by the school are identical:
\begin{claim}\label{clm: sIGM_equivalence}
sIGM holds if and only if for all $t \in \reals$:
\begin{align*}
\Pr \left[T_1 = t  | A_1 = 1 \right] = \Pr \left[T_2 = t  | A_2 = 1 \right]
\end{align*}
\end{claim}

\begin{proof}
See Appendix~\ref{app: sIGM_equivalence}
\end{proof}

\section{Inference Preliminaries}
In this section, we derive some basic properties of the joint distributions on student types, exam scores, admissions rules, and grades that are relevant for reasoning about the employer's Bayesian inference task. We will draw upon these basic results in the coming sections.

\subsection{Preliminaries on Gaussians and Multivariate Gaussians}

First, we observe that together, student types, exam scores, and grades are distributed according to a multi-variate Gaussian.
\begin{claim}
$(T_i,S_i,G_i)$ follows a multivariate normal distribution.
\end{claim}

\begin{proof}
A set of random variables is distributed according to a multivariate normal distribution if every linear combination of the variables is distributed as a univariate normal distribution.
For all $a,b,c \in \reals$, $a T_i + b S_i + c G_i = (a+b+c) T_i + b X_i + c Y_i$ follows a normal distribution as the sum of independent normal random variables.
\end{proof}
We now quote a basic fact about the conditional distribution that results when one starts with a multi-variate normal distribution, and conditions on the realization of a subset of its coordinates.
\begin{claim}\label{clm: conditional_MVN}
Let $n \geq 2$ be an integer. Let $Z \in \mathbb{R}^n$ be a random variable following a multi-variate normal distribution. Let $Z = (Z_1,Z_2)$ where $Z_i \in \mathbb{R}^{n_i}$ with $n_1 + n_2 = n$. Suppose $Z$ has mean $m =(m_1,m_2)$ where $m_i \in \reals^{n_i}$, and covariance matrix
\[
\Sigma =
\left[
\begin{array}{c|c}
\Sigma_{11} & \Sigma_{12}  \\
\hline
\Sigma_{21} & \Sigma_{22}
\end{array}
\right]
\]
where $\Sigma_{ij} \in \reals^{n_i \times n_j}$. Then $\E \left[ Z_1 | Z_2 = z_2 \right] = m_1 + \Sigma_{12} \Sigma_{22}^{-1} (z_2 - m_2)$ and $\Var \left[ Z_1 | Z_2 = z_2 \right] = \Sigma_{11} - \Sigma_{12} \Sigma_{22}^{-1} \Sigma_{21}$.
\end{claim}

\begin{proof}
See lecture notes \cite{MVN2008}.
\end{proof}

The following technical lemma will also be useful for us.
\begin{claim}\label{clm: hazard_rate}
The hazard rate $H(x) = \frac{\phi(x)}{1-\Phi(x)}$ of a  standard normal random variable is increasing, and satisfies
\[
\lim_{x \to -\infty} H(x) = 0,~H(x) = x + o_{x \to +\infty}(1)
\]
\end{claim}

This is a commonly known result in the literature on probability theory and statistics. For completeness, we provide a proof in Appendix \ref{app: hazard_rate}.

\subsection{Employer's First Moment Inference}

The main lemma of this section characterizes the employer's Bayesian inference task when the college is using a threshold admissions rule: the posterior expectation of a student's type, conditioned on their exam score being sufficiently high to cross the admissions threshold, and on their observed grade.  In the appendix, we give the corresponding inference rule for the employer when the college can use an arbitrary monotone admissions rule.
\begin{lemma}\label{lem: closed_form_bayes_posterior}
\begin{align*}
&\E \left[ T_i | S_i \geq \beta_i, G_i = g \right]
\\&= \frac{\gamma^2}{\sigma_i^2 + \gamma^2} \mu_i
+ \frac{\sigma_i^2}{\sigma^2 + \gamma^2} g
\\&+ \frac{\gamma^2 \sigma_i^2}{\sqrt{(\sigma_i^2+\gamma^2) (\sigma_i^2 + \gamma^2 + \gamma^2 \sigma_i^2)}}
\cdot
H \left( \frac{(\sigma_i^2 + \gamma^2) \cdot \beta_i - \gamma^2 \mu_i - \sigma_i^2 g}{ \sqrt{(\sigma_i^2+ \gamma^2) (\sigma_i^2+ \gamma^2 + \gamma^2 \sigma_i^2)}} \right)
\end{align*}
where $H(x) = \frac{\phi(x)}{1-\Phi(x)}$ is the Hazard function of a standard normal random variable.
\end{lemma}

\begin{proof}
The proof is given in Appendix~\ref{app: closed_form_bayes_posterior}.
\end{proof}

A corollary of the previous lemma is that the posterior expectation computed by the employer will satisfy a number of nice regularity conditions which will be useful in proving our impossibility results:
\begin{corollary}\label{cor: increasing_limits}
$e_i(\mu_i,\sigma_i,\beta_i,g) = \E \left[ T_i | S_i \geq \beta_i, G_i = g \right]$ is continuous, differentiable, and strictly increasing in each of $\mu_i,~g$ and $\beta_i$. Further,
\begin{align*}
&\lim_{g \to -\infty} e(\mu_i,\sigma_i,\beta_i,g) = -\infty,
\\&  \lim_{g \to +\infty} e_i(\mu_i,\sigma_i,\beta_i,g) = +\infty,
\end{align*}
and
\begin{align*}
&\lim_{\beta_i \to -\infty} e(\mu_i,\sigma_i,\beta_i,g) = \frac{\gamma^2}{\sigma_i^2 + \gamma^2} \mu_i
+ \frac{\sigma_i^2}{\sigma^2 + \gamma^2} g,
\\
&\lim_{\beta_i \to +\infty} e(\mu_i,\sigma_i,\beta_i,g) = + \infty.
\end{align*}
\end{corollary}

\begin{proof}
See Appendix \ref{app: increasing_limits}
\end{proof}

Finally, we define a quantity that will be useful to make reference to in a number of our forthcoming arguments: the minimum grade that results in a student from group $i$ being hired by the employer, given a fixed admissions rule.
\begin{definition}[Hiring threshold on grades]\label{def: g*}
We define $g_i^*(C) = \min \{g:~\E \left[ T_i | S_i \geq \beta_i, G_i = g \right] \geq C\}$ the inverse function of $g \to \E \left[ T_i | S_i \geq \beta_i, G_i = g \right]$.\
\end{definition}
By Corollary~\ref{cor: increasing_limits}, $g^*_i(.)$ is a well-defined function on domain $\reals$, and is continuous, differentiable, and strictly increasing.

\subsection{Moments of the posterior distribution for monotone admission rules}

The following lemma holds for the general case of monotone, randomized admission rules, and is useful in characterizing the moments of the distribution of types conditional on $A_i = 1$ and $G=g$ in population $i$:
\begin{lemma}\label{lem: conditional_derivatives}
Let $A_i(.)$ be a non-decreasing, non-zero, possibly randomized admission rule. For all $g \in \mathbb{R}$, $\E \left[T_i^k| G_i=g, A_i = 1 \right]$ is finite and differentiable in $g$, and its derivative satisfies the following equation:
\begin{align*}
&\frac{\partial}{\partial g} \E_i \left[T_i^k \middle\vert G_i=g, A_i = 1 \right]
\\&= \frac{1}{\gamma^2}  \E_i \left[T_i^{k+1}\middle\vert G_i = g, A_i=1\right]
\\&- \frac{1}{\gamma^2}  \E_i \left[T_i^k \middle\vert G_i =g, A_i=1 \right] \cdot \E_i \left[T_i |G_i=g,A_i=1 \right].
\end{align*}
\end{lemma}

\begin{proof}
The proof is given in Appendix~\ref{app: conditional_derivatives}.
\end{proof}

\section{When Both Conditions are Satisfiable}
In this section, we observe that there are two settings in which it is possible to ``have it all'' --- satisfying both IGM and equal opportunity even in the multiple threshold case. The first setting is that of noiseless exam scores: when student types are perfectly observable by the school. The second setting is when the school opts not to report grades. We view the first setting as generally unrealisable, since any student evaluation will involve some degree of stochasticity. However the 2nd case --- in which a school opts not to report grades --- can be realized.
\subsection{Noiseless Exam Scores (Observable Types)}

First, we observe that if schools can perfectly observe student types (we have noiseless exam scores with $S_i = T_i$), then there is a simple threshold admissions rule that simultaneously achieves IGM and equal opportunity, even in the multiple threshold case. The ideas is simple: Given a range of employer costs $[C^-, C^+]$, the college simply sets an admissions threshold of $C^+$ or higher, using the same threshold for members of both groups. Because the threshold is the same for both groups, the probability of being admitted to college is a function only of type, and independent of group membership conditioned on type. Because scores were noiseless, admissions to college deterministically certifies that a student's type $t_i \geq C^+$, and so the employer chooses to hire everyone, independently of the grade they receive (and independently of their group membership). Hence, the probability of being hired is the same as the probability of being accepted to college, and is independent of group membership conditioned on type, and the employer's hiring rule is independent of group membership.

\begin{claim}\label{clm: fairness_nonoise}
Suppose $S_i = T_i$, i.e. a student's score perfectly reveals his type. Then for any hiring interval of hiring costs $[C^-, C^+] \in \reals$, the non-zero admissions rule:
$$A_i(s) = 1 \Leftrightarrow s \geq C^+$$
 for both groups $i \in \{1,2\}$ satisfies IGM and equal opportunity when paired with \emph{any} grading policy.
\end{claim}

\begin{proof}
See Appendix \ref{app: fairness_nonoise}.
\end{proof}



\begin{claim}
Suppose the school does not assign grades to students. Then for any hiring interval of hiring costs $[C^-, C^+] \in \reals$, the non-zero thresholding admissions rule:
$$A_i(s) = 1 \Leftrightarrow s \geq \beta$$
 for both groups $i \in \{1,2\}$ satisfies IGM and equal opportunity
 when $\beta$ is large enough.
\end{claim}

\begin{proof}
For $\beta$ big enough, $\E\left[T_i | S_i \geq \beta \right] \geq C^+$ as
\\$\lim_{\beta \to +\infty} \E \left[T_i | S_i \geq \beta\right]= +\infty$; this can be seen either by following the same steps as in the proof of Lemma \ref{lem: closed_form_bayes_posterior} to obtain that
\[
\E\left[T_i | S_i \geq \beta \right] = \mu_i + \frac{\sigma_i^2}{\sqrt{1+\sigma_i^2}} H \left(\frac{\beta_i-\mu_i}{\sqrt{1+\sigma_i^2}} \right)
\]
which tends to $+\infty$ when $\beta_i \to +\infty$ by Claim \ref{clm: hazard_rate}. Another way of deriving this expression is by noting that not having a grade is equivalent to having an uninformative grade, i.e. to having $\gamma \to +\infty$. 
Now, let $\beta$ be large enough such that in both populations, such that $\E\left[T_i | S_i \geq \beta \right] \geq C^+$. IGM immediately holds as every student that is accepted by the school is hired by the employer. Equal opportunity holds because the probability of a student with type $t$ being hired by the employer is exactly the probability that he is admitted by the school (every student admitted by the school is hired by the employer), hence is given by
\[
\Pr \left[ S_i \geq \beta | T_i = t\right] = \int_{s \geq \beta} \phi(s-t) dt,
\]
and is independent of the student's population.
\end{proof}

Note that this result is achieved by having the school set a very high admissions threshold (uniformly for both groups), and declining to give grades. Hence, declining to give grades may be a reasonable strategy for promoting our fairness goals in a highly selective school, but does not work when admissions thresholds must be lower. We note that the practice of grade witholding in MBA programs seems to be limited to the very top programs \cite{nogrades}.

In the remainder of the paper we consider the case in which exam scores have positive finite variance, and in which the college uses a grading policy with positive finite variance. What will be possible will depend on whether we are in the single or multiple threshold case.

\section{The Single Threshold Case}
In this section, we consider what is possible when there is only a single employer with a hiring cost $C$ that is known to the college. We show that in this case, IGM can always be achieved, but that it is impossible to achieve sIGM.


\subsection{IGM can always be achieved}
The main idea is as follows: For any grading scheme, and with a single threshold $C$ in mind, the college can separately set different admissions thresholds $\beta_1^*$ and $\beta_2^*$ for the two groups respectively such that the posterior expectation for a student type from each group crosses the threshold of $C$ at a grade $g^*$, which can be made to be the same for both populations. Since the only thing that matters in the employer's hiring decision is whether or not the student's expected type is above or below $C$, this is enough to cause the employer's hiring decision to be independent of group membership. The next lemma establishes that it is always possible to find such thresholds:

\begin{lemma}\label{lem: weak_calibration}
For any $C$ in $\reals$, there exists thresholds $\beta_1^*$ and $\beta_2^*$ and a grade $g^*$ such that
\[
\E \left[T_1 | G_1 = g^*, S_1 \geq \beta_1^* \right] = \E \left[T_2 | G_2 = g^*, S_2 \geq \beta_2^* \right] = C
\]
\end{lemma}

\begin{proof}
It follows by Corollary~\ref{cor: increasing_limits} that
\[
\E \left[T_i | G_i = g, S_i \geq \beta_i \right]
\]
is continuous in $\beta_i$ and must reach any value between $\frac{\gamma^2}{\sigma_i^2 + \gamma^2} \mu_i + \frac{\sigma_i^2}{\sigma_i^2 + \gamma^2} g$ and $+\infty$. For $g^*$ small enough, it must be the case that
\[
\frac{\gamma^2}{\sigma_i^2 + \gamma^2} \mu_i + \frac{\sigma_i^2}{\sigma_i^2 + \gamma^2} g^* \leq C < +\infty,
\]
hence there exists $\beta_i^*$ such that
\[
\E \left[T_i | G_i = g^*, S_i \geq \beta_i^* \right] = C.
\]
\end{proof}

\begin{corollary}
Fix any $C$ in $\reals$. When the school uses thresholding admission rules with thresholds $\beta_1^*$ and $\beta_2^*$, IGM holds for that $C$.
\end{corollary}

\begin{proof}
$\E \left[T_i | G_i = g, S_i \geq \beta_i^* \right]$ is a strictly increasing function of $g$ by Corollary~\ref{cor: increasing_limits} , therefore the employer accepts students from any population if and only if $g \geq g^*$ where $g^*$ is population-independent, which proves the results.
\end{proof}

\subsection{sIGM is impossible}
We now show that strong IGM --- making the posterior distributions for both groups identical --- is impossible. In addition to its intrinsic interest, this result will be a key ingredient in our impossibility results for the multiple threshold setting.
\begin{lemma}\label{lem: strong_calib_gaussian}
Suppose the priors are distinct. For any two thresholds $\beta_1$ and $\beta_2$, there must exists $t \in \reals$ such that
\[
\Pr \left[T_1 = t | S_1 \geq \beta_1 \right] \neq \Pr \left[T_2 = t | S_2 \geq \beta_2 \right]
\]
I.e., sIGM cannot hold. 
\end{lemma}

\begin{proof}
Let $x_i(t) = \Pr \left[ S_i \geq \beta_i | T_i = t\right]$. Suppose for all $t \in \reals$, sIGM holds, i.e.
\[
\Pr \left[T_1 = t | S_1 \geq \beta_1 \right] \neq \Pr \left[T_2 = t | S_2 \geq \beta_2 \right]
\]
by Claim~\ref{clm: sIGM_equivalence}. Then
\[
\frac{x_1(t) \phi\left(\frac{t-\mu_1}{\sigma_1}\right)}{\Pr \left[S_1 \geq \beta_1  \right]} = \frac{x_2(t) \phi \left(\frac{t-\mu_2}{\sigma_2}\right)}{\Pr \left[S_2 \geq \beta_2 \right]}
\]
hence
\[
\frac{x_1(t)}{x_2(t)} = \frac{\sigma_1 \Pr \left[S_1 \geq \beta_1 \right]}{\sigma_2 \Pr \left[S_2 \geq \beta_2 \right]} \cdot \exp\left( \frac{(t-\mu_2)^2}{2\sigma_2}  - \frac{(t-\mu_1)^2}{2\sigma_1^2} \right)
\]
$x_1(.)$ and $x_2(.)$ are non-decreasing functions with values in $[0,1]$, and $x_i(t) = \int_{s \geq \beta_i} \phi(s-t) ds$ is non-zero; therefore, $\lim_{t = +\infty} x_i(t)$ exists and is strictly positive. It must then be the case that $\frac{x_1(t)}{x_2(t)}$ has a finite and strictly positive limit in $+\infty$. On the other hand,
\[
\exp\left( \frac{(t-\mu_2)^2}{2\sigma_2}  - \frac{(t-\mu_1)^2}{2\sigma_1^2} \right) = K \exp \left(\frac{t^2}{2} \left(\frac{1}{\sigma_2^2}-\frac{1}{\sigma_1^2}\right) + \left( \frac{\mu_1}{\sigma_1^2} - \frac{\mu_2}{\sigma_2^2}\right) t \right)
\]
for some constant $K$. It is easy to see that the above quantity tends to either $+\infty$ or $0$ as $t \to +\infty$ as long as either $\sigma_1 \neq \sigma_2$ or $\mu_1 \neq \mu_2$ (one of $\frac{1}{\sigma_2^2}-\frac{1}{\sigma_1^2}$ and $\frac{\mu_1}{\sigma_1^2} - \frac{\mu_2}{\sigma_2^2}$ must be non-zero). This leads to a contradiction.
\end{proof}

\subsection{Equal opportunity}


We defer the technical results of this section to Appendix~\ref{app: equalodd_partial}.
Lemma~\ref{lem: equalodd_partial_1} shows that for thresholding admission rules, IGM and equal opportunity cannot simultaneously hold for Gaussian priors with the same variance but different mean. This shows that obtaining fairness in the general case is significantly more difficult than in the simple cases in which the types are observable and the school does not assign grades.
Lemma~\ref{lem: equalodd_partial_2} shows that arguably stringent conditions on the grade accuracy and the thresholds set by the school must hold for equal opportunity to be possible. We conjecture that these conditions are, in general, impossible to satisfy, making equal opportunity impossible to satisfy even in isolation, in the single threshold case. As we will see in the next section, it is impossible to satisfy in the multiple-threshold case.

\section{The Multiple Threshold Case}
In this section, we turn to the multiple threshold case, which we view as the main setting of interest. In this case, we ask whether we can achieve IGM and equal opportunity not just with respect to a single known hiring cost $C$, but with respect to an entire interval of hiring costs $C \in [C^-, C^+]$. This will be the case when there are multiple employers, or simply when there is some uncertainty about the hiring threshold used by a single employer.

\subsection{IGM is Impossible}
In this section, we show that IGM is impossible to achieve even in isolation. The proof proceeds by showing that in the multiple threshold case, IGM must imply sIGM --- i.e. that the posterior distributions conditional on admission to college are identical for both groups. Impossibility then follows from the impossibility of achieving sIGM (even for a single threshold), which we proved in the last section.

We first state a technical lemma, showing that if we satisfy IGM for every employer cost $C$ in a continuous interval, we must actually be equalizing the posterior expected type across groups for every grade $g$ in some other continuous interval.

\begin{claim}\label{clm: IGM_threshold_to_grade}
Let $\ubar{C} < \bar{C}$. Suppose that for all $C \in (\ubar{C},\bar{C})$,
\[
\E\left[T_1|G_1=g, S_1 \geq \beta_1\right] \geq C \Leftrightarrow \E\left[T_2|G_2=g, S_2 \geq \beta_2\right] \geq C,
\]
then it must be the case that for $g$ in some interval $(a,b)$,
\[
\E\left[T_1|G_1=g, S_1 \geq \beta_1\right] = \E\left[T_2|G_2=g, S_2 \geq \beta_2\right]
\]
\end{claim}

\begin{proof}
Let $a = g_1^*(\ubar{C})$ and $b = g_1^*(\bar{C})$ where $g_1^*(.)$ is the strictly increasing inverse of $g \to \E \left[T_1 | G_1 = g, A_1 = 1 \right]$ as per Claim \ref{cor: increasing_limits} and Definition \ref{def: g*}. Suppose there exists $g \in (a,b)$ such that
\[
\E\left[T_1|A_1=1,G_1=g\right] > \E\left[T_2|A_2=1,G_2=g\right],
\]
then for $C = \E\left[T_1|A_1=1,G_1=g\right] \in (\ubar{C},\bar{C})$, it must be the case that
\[
\E\left[T_1|A_1=1,G_1=g\right] \geq C > \E\left[T_2|A_2=1,G_2=g\right]
\]
which contradicts the assumption of the claim. Now, suppose there exists $g \in (a,b)$ such that
\[
\E\left[T_1|A_1=1,G_1=g\right] < \E\left[T_2|A_2=1,G_2=g\right],
\]
then for $\epsilon > 0$ small enough, $C = \E\left[T_1|A_1=1,G_1=g\right]+\epsilon \in (\ubar{C},\bar{C})$ and
\[
\E\left[T_1|A_1=1,G_1=g\right] < C \leq \E\left[T_2|A_2=1,G_2=g\right]
\]
which also contradicts the assumption of the claim. Therefore, it must be the case that for all $g \in (a,b)$,
\[
\E\left[T_1|A_1=1,G_1=g\right] = \E\left[T_2|A_2=1,G_2=g\right]
\]
\end{proof}

We can now go on to prove the main theorem in this section:
\begin{theorem}\label{thm: robust_IGM}
Suppose the priors are distinct,
 then IGM cannot for all hiring costs $C \in (\ubar{C},\bar{C})$.
\end{theorem}

\begin{proof}
By Claim \ref{clm: IGM_threshold_to_grade}, it must be the case that for all $g \in (a,b)$ for some interval $(a,b)$, $\E\left[T_1|G_1=g, S_1 \geq \beta_1\right] = \E\left[T_2|G_2=g, S_2 \geq \beta_2 \right]$. For all $g$ in $(a,b)$, by Lemma~\ref{lem: conditional_derivatives}
\[
\frac{\partial}{\partial g} \E \left[T_1 | G_1=g, S_1 \geq \beta_1 \right] = \frac{\partial}{\partial g} \E \left[T_2 | G_2=g, S_2 \geq \beta_2 \right]
\]
and hence $E\left[T_1^2|G_1=g, S_1 \geq \beta_1\right] = \E\left[T_2^2|G_2=g, S_2 \geq \beta_2\right]$. It is easy to see that using the same argument by induction yields that for all integers $k$,
\[
\E\left[T_1^k|G_1=g, S_1 \geq \beta_1\right] = \E\left[T_2^k|G_2=g,S_1 \geq \beta_1\right]
\]
Since the distributions of types for the two populations conditional on $G_i=g, S_i \geq \beta_i$ admit a moment generating function (this follows immediately from the fact that $P_i$ admits a moment generating function) and have identical moments, it must be that the distributions are the same for all $g \in (a,b)$. I.e., for all $g \in (a,b)$, we have
\[
\Pr \left[T_1 = t |G_1=g, S_1 \geq \beta_1\right] = \Pr \left[T_2 = t |G_2=g, S_1 \geq \beta_1\right]
\]
We have that in population $i$,
\[
\Pr \left[T_i = t | G_i=g, S_i \geq \beta_i \right] = \frac{\Pr \left[T_i = t | S_i \geq \beta_i \right]  \phi \left(\frac{g-t}{\gamma} \right)}{\int_t \Pr \left[T_i = t | S_i \geq \beta_i \right] \phi \left(\frac{g-t}{\gamma} \right) dt}
\]
Note that $\int_t \Pr \left[T_i = t | S_i \geq \beta_i \right] \phi \left(\frac{g-t}{\gamma} \right) dt$ is a function of $g$ only, that we will denote $p_i(g)$ from now on.
\[
\Pr \left[T_1 = t | G_1=g, S_1 \geq \beta_i \right] = \Pr \left[T_2 = t | G_2=g, S_2 \geq \beta_2 \right]
\]
implies
\[
\frac{\Pr \left[T_1 = t | S_1 \geq \beta_1 \right]}{\Pr \left[T_2 = t | S_2 \geq \beta_2 \right]} = \frac{p_1(g)}{p_2(g)}
\]
for all $g \in (a,b)$ and $t \in \reals$. $\Pr \left[T_1 = t | S_1 \geq \beta_1 \right]$ and $\Pr \left[T_2 = t | S_2 \geq \beta_2 \right]$ are both probability density functions that integrate to $1$, so it must be the case that $\frac{p_1(g)}{p_2(g)} = 1$ and $\Pr \left[T_1 = t | S_1 \geq \beta_1 \right] = \Pr \left[T_2 = t | S_2 \geq \beta_2 \right]$. Therefore, sIGM must hold, which we have shown is impossible in Lemma \ref{lem: strong_calib_gaussian}.
\end{proof}

\subsection{Equal opportunity cannot hold}
Finally, we show that in the multiple threshold case, it is also impossible to satisfy the equal opportunity condition.

\begin{theorem}\label{thm: impossibility_equalodds}
Suppose the priors are distinct. There exist no thresholding admission rules such that equal opportunity is guaranteed for all $C \in (\ubar{C},\bar{C})$, for any $\ubar{C} < \bar{C}$.
\end{theorem}

\begin{proof}
It is easy to see $x_i(t) = \int_s A_i(s) \phi(s-t) ds = \int_u A_i(u+t) \phi(u) du$ is monotone non-decreasing in $t$ and non-zero. Remember
\[
e_i(g) = \E \left[T_i | G_i = g, S_i \geq \beta_i \right]
\]
has a strictly increasing and differentiable inverse $g^*_i(.)$ on $(-\infty,+\infty)$ by Corollary~\ref{cor: increasing_limits}, and a student is hired by the employer if and only if $g \geq g^*_i(C)$. A student with type $t$ in population $i$ gets therefore hired with probability
\[
\int_{g \geq g^*(C)} x_i(t) \phi \left(\frac{g-t}{\gamma} \right) dt = x_i(t) \left( 1 - \Phi \left(\frac{g^*_i(C)-t}{\gamma} \right) \right)
\]
equal opportunity then imply that $\forall t \in \reals, C \in (\ubar{C},\bar{C})$,
\[
\frac{x_1(t)}{x_2(t)} \cdot  \left( 1 - \Phi \left(\frac{g^*_1(C)-t}{\gamma} \right) \right)
=  \left( 1 - \Phi \left(\frac{g^*_2(C)-t}{\gamma} \right) \right)
\]
Taking the first order derivative in $C$ of both sides of the above equation, we have that for all $C \in (\ubar{C},\bar{C})$, for all $t \in \reals$,
\[
\frac{x_1(t)}{x_2(t)} \cdot \frac{\frac{\partial g_1^*}{\partial C}(C)}{\frac{\partial g_2^*}{\partial C}(C)} = \frac{\phi \left(\frac{g^*_2(C)-t}{\gamma}\right)}{\phi \left( \frac{g^*_1(C)-t}{\gamma}\right)}
\]
Suppose for some $C \in (\ubar{C},\bar{C})$, $g^*_1(C) \neq g^*_2(C)$. Without loss of generality, renumber the populations such that $g^*_2(C) > g^*_1(C)$. We have that
\[
\frac{\phi \left(\frac{g^*_2(C)-t}{\gamma}\right)}{\phi \left( \frac{g^*_1(C)-t}{\gamma}\right)} = \exp \left(\frac{2 \left( g_2^*(C) - g_1^*(C)\right) t + g_1^*(C)^2-g_2^*(C)^2}{2 \gamma^2} \right)
\]
and we know that $g_i^*(.)$ is a strictly increasing function so $\frac{\partial g_2^*}{\partial C}(C) > 0$ so it must be the case that
\[
\lim_{t \to +\infty} \frac{x_1(t)}{x_2(t)} = +\infty.
\]
Since $x_1(t)$ is upper-bounded by $1$, this implies in particular that $x_2(t) \to 0$ as $t \to +\infty$, which contradicts $x_2(.)$ being a non-zero, non-decreasing function. Hence, it must be the case that for all $C \in (\ubar{C},\bar{C})$, $g_1^*(C) = g_2^*(C)$, i.e. IGM holds. By Lemma~\ref{thm: robust_IGM}, this is impossible.
\end{proof}

\section{Conclusion}
We consider two natural fairness goals that a college might have for its affirmative action policies: granting equal opportunity to individuals with the same type when graduating from high school, independent of their group membership, and incentivizing downstream employers to make hiring decisions that are independent of group membership. We show that these goals can be simultaneously achieved by highly selective colleges (i.e. those with very high admissions thresholds) --- \emph{but only if they do not report grades to employers}. This provides another view on this practice, which is followed by several highly selective MBA programs. On the other hand, we find that these goals are generally unachievable even in isolation if schools report informative grades. These impossibility results crucially hinge on the fact that exam scores and grades provide only \emph{noisy} signals about student types, and hence require rational expectation maximizers to reason about prior type distributions, which can vary by group.

Our paper leaves open a natural technical question: can a college set admissions and informative grading policies to realize the equal opportunity condition, in the \emph{single threshold} case? We conjecture that the answer to this question is \emph{no}, and in the Appendix, we give a theorem supporting this conjecture --- ruling out the possibility for deterministic admissions rules in every case except when the grading variance is exactly $1$.

\subsection*{Acknowledgements}
We thank Mallesh Pai and Jonathan Ullman for helpful discussions at an early stage of this work.

\bibliographystyle{plainnat}
\bibliography{refs}

\appendix

\section{Omitted proofs}

\subsection{Proof of Claim \ref{clm: sIGM_equivalence}}\label{app: sIGM_equivalence}

\begin{align*}
&\Pr \left[T_i = t  | G_i = g, A_i = 1 \right]
\\&= \frac{\Pr \left[\left(G_i = g | T_i = t \right) |  A_i = 1 \right] \cdot \Pr \left[T_i = t  | A_i = 1 \right]}{\Pr\left[G_i = g | A_i = 1\right]}
\end{align*}
Remembering that $G_i = T_i + Y$ where $Y \sim \mathcal{N} \left(0,\gamma^2 \right)$, we have
\[
\Pr \left[\left( G_i = g | T_i = t \right) |  A_i = 1 \right] = \Pr \left[ Y = g-t \right]= \phi \left(\frac{g-t}{\gamma} \right),
\]
does not depend on the population $i$, hence sIGM holds if and only if
\[
\Pr \left[T_1 = t  |  A_1 = 1 \right] = \frac{\Pr\left[G_1 = g | A_1 = 1\right]}{\Pr\left[G_2 = g | A_2 = 1\right]} \cdot \Pr \left[T_2 = t  |  A_2 = 1 \right]
\]
$\Pr \left[T_1 = t  |  A_1 = 1 \right]$ and $\Pr \left[T_2 = t  |  A_2 = 1 \right]$ are probability density functions over $t$ that both integrate to $1$, and $\frac{\Pr\left[G_1 = g| A_1=1\right]}{\Pr\left[G_2 = g|A_2=1\right]}$ is constant in $t$, therefore the above equation holds if and only if
\[
\Pr \left[T_1 = t  |  A_1 = 1 \right] = \Pr \left[T_2 = t  |  A_2 = 1 \right].
\]
and
\[
Pr\left[G_1 = g | A_1 = 1\right] = Pr\left[G_2 = g | A_2 = 1\right].
\]
This is equivalent to simply
\[
\Pr \left[T_1 = t  |  A_1 = 1 \right] = \Pr \left[T_2 = t  |  A_2 = 1 \right],
\]
noting that when the posterior over types after admission by the school are equal, the distribution of grades must be equal too, as the distribution of $Y$ is population-independent.

\subsection{Proof of Claim \ref{clm: hazard_rate}}\label{app: hazard_rate}

We start by showing bounds on $H(x)$. First, we note that
\begin{align*}
1 - \Phi(x)
&= \int_{t = x}^{+\infty} \frac{1}{\sqrt{2\pi}} \exp\left(-t^2/2\right) dt
\\&\leq \int_{t = x}^{+\infty} \frac{1}{\sqrt{2\pi}} \frac{t}{x} \exp\left(-t^2/2\right) dt
\\& = \frac{1}{x \sqrt{2\pi}} \exp\left(-x^2/2\right)
\\&= \frac{\phi(x)}{x}
\end{align*}
where the inequality follows from $t/x \geq 1$ for $t \geq x$. This immediately implies that $H(x) \geq x$. Now, we show that
\[
g(x) = 1 - \Phi(x) -  \frac{x}{x^2+1} \phi(x) \geq 0
\]
This follows from noting that
\begin{align*}
g'(x)
&= - \phi(x) - \frac{1 - x^2}{(x^2+1)^2} \phi(x) + \frac{x^2}{x^2+1} \phi(x)
\\&= \frac{x^4 + x^2 +x^2 - 1 - x^4 -2 x^2 - 1}{(x^2+1)^2} \phi(x)
\\&= - \frac{2}{(x^2+1)^2} \phi(x)
\end{align*}
so $g$ is a decreasing function with $\lim_{x \to +\infty} g(x) = 0$. Therefore,
\[
H(x) \leq \frac{x^2+1}{x} = x + \frac{1}{x}
\]
Now, $\lim_{x \to -\infty} H(x) = 0$ as $\phi(x) \to 0$ and $1-\Phi(x) \to 1$. $x \leq H(x) \leq x + \frac{1}{x}$ is enough to show that $H(x) = x + o(x)$ as $x \to +\infty$. Finally,
\begin{align*}
H'(x)
& = \frac{-x \phi(x) \left(1-\Phi(x)\right) + \phi(x)^2}{\left(1 - \Phi(x)\right)^2}
\\& =  - x H(x) + H(x)^2
\\&= H(x) \left( H(x) - x\right)
\\& \geq 0
\end{align*}
as $H(x) \geq 0$ and $H(x) \geq x$. This concludes the proof.

\subsection{Proof of Lemma~\ref{lem: closed_form_bayes_posterior}}\label{app: closed_form_bayes_posterior}

For simplicity of notations, we drop the $i$ subscripts in the whole proof. We let  $Z_1 = T$, $Z_2 = (S,G)$ and $z_2 = (s,g) \in \reals^2$, and apply Claim \ref{clm: conditional_MVN} with $m_1 = \E [T] = \mu$, $m_2=[\mu,\mu]^\top$, and
\[
\Sigma =
\left[
\begin{array}{c c c}
\sigma^2 & \sigma^2 & \sigma^2 \\
\sigma^2 & \sigma^2 + 1 & \sigma^2\\
\sigma^2 & \sigma^2 & \sigma^2 + \gamma^2 \\
\end{array}
\right]
\]
and we have $\Sigma_{12} = [\sigma^2 \; \sigma^2]$ and $\Sigma_{22} = \left[\begin{array}{c c}
\sigma^2 + 1 & \sigma^2\\
\sigma^2 & \sigma^2 + \gamma^2
\end{array}
\right]
$.  Claim \ref{clm: conditional_MVN} yields
\begin{align*}
\mathbb{E} \left[ T | S = s, G = g\right]
&= \mu + \frac{1}{\sigma^2 + \gamma^2 + \sigma^2 \gamma^2} \left(\sigma^2 \gamma^2 (s-\mu) + \sigma^2 (g-\mu)\right)
\\&= \frac{\gamma^2}{\sigma^2 + \gamma^2 + \sigma^2 \gamma^2} \mu +\frac{\sigma^2}{\sigma^2 + \gamma^2 + \sigma^2 \gamma^2} g
\\&+ \frac{\gamma^2 \sigma^2}{\sigma^2 + \gamma^2 + \sigma^2 \gamma^2} s.
\end{align*}
Therefore, it must be the case that
\begin{align*}
&\mathbb{E} \left[ T | S \geq \beta , G = g\right]
\\&= \int_{s} \mathbb{E} \left[ T | S = s, G = g\right] \Pr \left[ S = s, G = g | S \geq \beta, G = g\right] ds
\\&= \int_{s \geq \beta} \mathbb{E} \left[ T | S = s, G = g\right] \Pr \left[ S = s | S \geq \beta, G = g\right] ds
\\& = \frac{\gamma^2}{\sigma^2 + \gamma^2 + \sigma^2 \gamma^2} \cdot \mu + \frac{\sigma^2}{\sigma^2 + \gamma^2 + \sigma^2 \gamma^2} \cdot g
\\&+ \frac{\gamma^2 \sigma^2}{\sigma^2 + \gamma^2 + \sigma^2 \gamma^2}  \int_{s \geq \beta} s \Pr \left[ S = s | S \geq \beta, G = g\right] ds
\\&= \frac{\gamma^2}{\sigma^2 + \gamma^2 + \sigma^2 \gamma^2} \cdot \mu + \frac{\sigma^2}{\sigma^2 + \gamma^2 + \sigma^2 \gamma^2} \cdot g
\\&+  \frac{\gamma^2 \sigma^2}{\sigma^2 + \gamma^2 + \sigma^2 \gamma^2} \mathbb{E} \left[ S | S \geq \beta , G = g\right]
\end{align*}
Because $(S,G)$ is a multivariate Gaussian, by Claim \ref{clm: conditional_MVN}, $S|G$ is a normal random variable with mean
\[
\mu + \frac{\sigma^2}{\sigma^2+\gamma^2} (g - \mu) = \frac{\gamma^2}{\sigma^2+\gamma^2} \cdot \mu + \frac{\sigma^2}{\sigma^2+\gamma^2} \cdot g
\]
and variance
\[
\sigma^2 + 1 - \frac{\sigma^4}{\sigma^2+ \gamma^2} = \frac{\sigma^2 \gamma^2 + \sigma^2+ \gamma^2}{\sigma^2+ \gamma^2}
\]
It immediately follows that
\begin{align*}
&\mathbb{E} \left[ S | S \geq \beta , G = g\right]
\\&= \frac{\gamma^2}{\sigma^2+\gamma^2} \cdot \mu
+ \frac{\sigma^2}{\sigma^2+\gamma^2} \cdot g
\\&+ \sqrt{\frac{\sigma^2 \gamma^2 + \sigma^2+ \gamma^2}{\sigma^2+ \gamma^2}}
\cdot H \left(
\frac{\beta - \frac{\gamma^2}{\sigma^2+\gamma^2} \cdot \mu - \frac{\sigma^2}{\sigma^2+\gamma^2} \cdot g}{\sqrt{\frac{\sigma^2 \gamma^2 + \sigma^2+ \gamma^2}{\sigma^2+ \gamma^2}}}
\right)
\\& = \frac{\gamma^2}{\sigma^2+\gamma^2} \cdot \mu + \frac{\sigma^2}{\sigma^2+\gamma^2} \cdot g
\\& + \sqrt{\frac{\sigma^2 \gamma^2 + \sigma^2+ \gamma^2}{\sigma^2+ \gamma^2}}
\cdot
H \left( \frac{(\sigma^2 + \gamma^2) \cdot \beta - \gamma^2 \mu - \sigma^2 g}{\sqrt{(\sigma^2+ \gamma^2) (\sigma^2+ \gamma^2 + \gamma^2 \sigma^2)}} \right)
\end{align*}
as it is the mean of Gaussian $S|G=g$ truncated at $\beta$. Therefore, we have that
\begin{align*}
&\E \left[ T | S \geq \beta, G = g \right]
\\&= \frac{\gamma^2}{\sigma^2+ \gamma^2+\gamma^2 \sigma^2} \left( 1 + \frac{\sigma^2 \gamma^2}{\sigma^2 + \gamma^ 2} \right) \mu
\\&+ \frac{\sigma^2}{\sigma^2+ \gamma^2+\gamma^2 \sigma^2} \left( 1 + \frac{\sigma^2 \gamma^2}{\sigma^2 + \gamma^ 2} \right) g
\\&+ \frac{\gamma^2 \sigma^2}{\sqrt{(\sigma^2+\gamma^2) (\sigma^2 + \gamma^2 + \gamma^2 \sigma^2)}}
\cdot H \left( \frac{(\sigma^2 + \gamma^2) \cdot \beta - \gamma^2 \mu - \sigma^2 g}{\sqrt{(\sigma^2+ \gamma^2) (\sigma^2+ \gamma^2 + \gamma^2 \sigma^2)}} \right)
\end{align*}
The term in front of $\mu$ is
\begin{align*}
&\frac{\gamma^2}{\sigma^2 + \gamma^2 + \sigma^2 \gamma^2} \left( 1 + \frac{\sigma^2 \gamma^2}{\sigma^2 + \gamma^ 2}  \right) \\&= \frac{\gamma^2}{\sigma^2 + \gamma^2 + \sigma^2 \gamma^2} \frac{\sigma^2 + \gamma^2 + \sigma^2 \gamma^2}{\sigma^2 + \gamma^ 2}
\\& = \frac{\gamma^2}{\sigma^2 + \gamma^2}
\end{align*}
Similarly, the term in front of $g$ is given by
\begin{align*}
\frac{\sigma^2}{\sigma^2 + \gamma^2 + \sigma^2 \gamma^2} \left( 1 + \frac{\sigma^2 \gamma^2}{\sigma^2 + \gamma^ 2} \right)
= \frac{\sigma^2}{\sigma^2 + \gamma^2}
\end{align*}

\subsection{Proof of Corollary \ref{cor: increasing_limits}}\label{app: increasing_limits}

Continuity and differentiability follow immediately from the closed-form expression for the mean, and the fact that the hazard rate is continuous and differentiable. The limits at $+\infty$ for $g$ and both limits for $\beta_i$ follow directly from the fact that $\lim_{x \to -\infty} H(x) = 0,~\lim_{x \to +\infty} H(x) = +\infty$ by Claim~\ref{clm: hazard_rate}. As $g \to -\infty$,
\[
\frac{(\sigma^2 + \gamma^2) \cdot \beta - \gamma^2 \mu - \sigma^2 g}{\sqrt{(\sigma^2+ \gamma^2) (\sigma^2+ \gamma^2 + \gamma^2 \sigma^2)}} \to +\infty,
\]
and as $x \to +\infty$, $H(x) = x + o \left( 1 \right)$ by Claim~\ref{clm: hazard_rate}. Therefore,
\begin{align*}
&e_i(\mu_i,\sigma_i,\beta_i,g)
\\& = \frac{\gamma^2}{\sigma_i^2 + \gamma^2} \mu_i
+ \frac{\sigma_i^2}{\sigma^2 + \gamma^2} g
\\&+ \frac{\gamma^2 \sigma_i^2}{\sqrt{(\sigma_i^2+\gamma^2) (\sigma_i^2 + \gamma^2 + \gamma^2 \sigma_i^2)}}
\cdot \frac{(\sigma_i^2 + \gamma^2) \cdot \beta_i - \gamma^2 \mu_i - \sigma_i^2 g}{\sqrt{(\sigma_i^2+ \gamma^2) (\sigma_i^2+ \gamma^2 + \gamma^2 \sigma_i^2)}}
\\& + o_{g \to -\infty}(1)
\end{align*}
The term in front of $g$ is given by
\[
\frac{\sigma_i^2}{\sigma^2 + \gamma^2} \left(1 - \frac{\gamma^2 \sigma_i^2}{\sigma_i^2+ \gamma^2 + \gamma^2 \sigma_i^2} \right) > 0
\]
hence the limit at $-\infty$ is $-\infty$. Monotonicity in $\beta_i$ is immediate from the fact that the hazard rate of a normal random variable is strictly increasing by Claim~\ref{clm: hazard_rate}. Finally, the mean of a truncated Gaussian is strictly increasing in the mean parameter of the Gaussian at constant variance. In particular, $S_i|G_i=g$ is Gaussian with mean $\frac{\gamma^2}{\sigma_i^2+\gamma^2} \mu + \frac{\sigma_i^2}{\sigma^2+\gamma^2} g$ and variance constant in both $g$ and $\mu$ by Claim~\ref{clm: conditional_MVN}; therefore, $\E \left[ S_i | G_i=g, S_i \geq \beta_i \right]$ is monotone strictly increasing in $\frac{\gamma^2}{\sigma_i^2+\gamma^2} \mu + \frac{\sigma_i^2}{\sigma_i^2+\gamma^2} g$, and hence in $g$ and in $\mu$. Because by the proof of Lemma~\ref{lem: closed_form_bayes_posterior},
\begin{align*}
\mathbb{E} \left[ T_i | S_i \geq \beta , G_i = g\right]
&= \frac{\gamma^2}{\sigma^2 + \gamma^2 + \sigma^2 \gamma^2} \mu_i + \frac{\sigma^2}{\sigma^2 + \gamma^2 + \sigma^2 \gamma^2} g
\\& +  \frac{\gamma^2 \sigma^2}{\sigma^2 + \gamma^2 + \sigma^2 \gamma^2} \mathbb{E} \left[ S_i | S_i \geq \beta , G_i = g\right],
\end{align*}
it immediately follow that $\E \left[ T_i| G_i=g, S_i \geq \beta\right]$ is strictly increasing in $g$ and in $\mu_i$.

\subsection{Proof of Lemma \ref{lem: conditional_derivatives}}\label{app: conditional_derivatives}

In the whole proof, we write $P_i(t) = \phi \left( \frac{t-\mu_i}{\sigma_i} \right)$ be the probability density function of a normal random variable with mean $\mu_i$ and variance $\sigma_i^2$, evaluated at $t$.

\begin{claim}\label{clm: cont}
$x_i(.)$ is continuous, increasing, and $x_i(t) > 0$ $\forall t \in \reals$.
\end{claim}

\begin{proof}
Note that
\begin{align*}
x_i(t) &=\int_s \Pr \left[A_i = 1 | S_i = s \right] \Pr\left[S_i=s | T_i=t\right] ds
\\&= \int_s A_i(s) \phi(s-t) ds
\\&= \int_u A_i(u+t) \phi(u) du
\end{align*}
Strict positivity and monotonicity follow immediately from the fact that $A_i(.)$ is a non-zero, non-decreasing admission rule. For any $h >0$, as $\left\vert A_i(s) \right\vert \leq 1$, we have that
\begin{align*}
\left\vert x_i(t+h) - x_i(t) \right\vert
&\leq \int_s \left\vert \phi(s-t-h) - \phi(s-t) \right\vert ds
\\&= \int_u \left\vert \phi(u-h) - \phi(u) \right\vert du
\end{align*}
As $\phi(u-h) - \phi(u) \geq 0$ iff $u \geq h/2$, we further have that
\begin{align*}
\left\vert x_i(t+h) - x_i(t) \right\vert
& \leq \int_{u \geq h/2} \left(\phi(u-h) - \phi(u) \right) du
\\& + \int_{u \leq h/2} \left(\phi(u) - \phi(u-h) \right) du
\\&= 2 \left( \Phi(h/2) - \Phi(-h/2) \right)
\end{align*}
Therefore, it follows that
\[
\lim_{h \to 0^+} \left\vert x_i(t+h) - x_i(t) \right\vert \to 0
\]
and by a similar argument that
\[
\lim_{h \to 0^-} \left\vert x_i(t+h) - x_i(t) \right\vert \to 0
\]
This concludes the proof.
\end{proof}

We can now proceed to the proof of Lemma~\ref{lem: conditional_derivatives}:

\begin{proof}
For simplicity of notations, let
\[
h_i(g,t) = \phi\left(\frac{g-t}{\gamma}\right) x_i(t) P_i(t)
\]
and remark that we can write
\[
\E_i \left[T_i^k | G_i=g, A_i = 1 \right] = \frac{\int_t t^k h_i(g,t) dt}{\int_t h_i(g,t) dt}.
\]
$x_i(t)> 0$ $\forall t \in \reals$ from Claim~\ref{clm: cont}, and $P_i(t) > 0$ for all $t \in \reals$, therefore $\int_t h_i(g,t) dt > 0~\forall g \in \reals$ and the above expectation is well-defined. For all $k$, $\int_t P_i(t) dt, \int_t |t|^k P_i(t) dt < +\infty$ because the moment generating function of a Gaussian exists and is finite, and $\left\vert x_i(t) \phi\left(\frac{g-t}{\gamma}\right)  \right\vert \leq \frac{1}{\sqrt{2\pi} \gamma}$, therefore $t h_i(g,t)$ and $h_i(g,t)$ are integrable for all $g$. We have that for all $t$, $h_i$ is differentiable everywhere in $g$, and that
\[
\frac{\partial h_i}{\partial g}(g,t) = \frac{t-g}{\gamma^2} \cdot \phi\left(\frac{g-t}{\gamma}\right) x_i(t) P_i(t).
\]
Since $x \rightarrow |x| \exp\left(-x^2/2\gamma^2\right) \leq \gamma e^{-1/2}$ for all $x$, we have that
\[
|t-g| \exp \left( -(t-g)^2/2\gamma^2 \right) \leq \gamma e^{-1/2}
\]
hence $\left\vert \frac{\partial h}{\partial g}(g,t) \right\vert \leq \frac{\gamma e^{-1/2}}{\sqrt{2\pi}\gamma} P_i(t)$ and $\left\vert t^k \cdot  \frac{\partial h}{\partial g}(g,t) \right\vert \leq \frac{\gamma e^{-1/2}}{\sqrt{2\pi}\gamma} t^k P_i(t)$. Since $P_i(t)$ and $|t|^k P_i(t)$ are integrable (because the moment generating function of a Gaussian exists and is finite), we can differentiate under the integral sign and show that
\[
\frac{\partial}{\partial g}\left[ \int_t h_i(g,t) dt \right] = \int_t \frac{\partial h_i}{\partial g}(g,t) dt
\]
and
\[
\frac{\partial}{\partial g}\left[ \int_t t^k \cdot h_i(g,t) dt \right] = \int_t t^k \frac{\partial h_i}{\partial g}(g,t) dt
\]

Therefore, $\E_i \left[T_i^k | G_i=g, A_i = 1 \right]$ is differentiable (hence continuous) in $g$ and has derivative
\begin{align*}
&\frac{\partial}{\partial g} \E_i \left[T_i^k | G_i=g, A_i = 1 \right]
\\&= \frac{1}{\gamma^2} \frac{\int_t t^k(t-g) h_i(g,t) dt \cdot \int_t h_i(g,t) dt}{\left(\int_t h_i(g,t) dt\right)^2}
\\&- \frac{1}{\gamma^2} \frac{\int_t t^k h_i(g,t) dt \cdot \int_t (t-g) h_i(g,t) dt}{\left(\int_t h_i(g,t) dt\right)^2}
\\&= \frac{1}{\gamma^2} \left(\frac{\int_t t^{k+1} h_i(g,t) dt}{\int_t h_i(g,t) dt} - \frac{\int_t t^{k} h_i(g,t) dt}{\int_t h_i(g,t) dt} \cdot \frac{\int_t t h_i(g,t) dt}{\int_t h_i(g,t) dt}\right)
\\& = \frac{1}{\gamma^2} \E_i \left[T_i^{k+1} \middle\vert A_i=1, G_i= g \right]
\\&- \frac{1}{\gamma^2}  \E_i \left[T_i^k \middle\vert A_i=1, G_i= g \right] \cdot \E_i \left[T_i |A_i=1, G_i= g \right]
\end{align*}
\end{proof}

\subsection{Proof of Claim \ref{clm: fairness_nonoise}}\label{app: fairness_nonoise}
Conditional on acceptance, $S_i \geq C^+$ and hence so $T_i \geq C^+$ with probability 1. Hence, for any grading policy and any grade $g$ it must be that
\[
\E \left[T_i | G_i = g, A_i = 1 \right] \geq C^+
\]
Therefore, the employer hires all students accepted by the school, and IGM holds. Because the probability of being hired by the employer is then just given by $A_i(t)$ for a student with type $t$, and because $A_1(t) = A_2(t)$, equal opportunity also holds.

\section{Hardness of equal opportunity in the single threshold case}\label{app: equalodd_partial}

\begin{lemma}\label{lem: equalodd_partial_1}
For any school hiring cost $C$, and any two Gaussians priors with different means $\mu_1 > \mu_2$ and same variances $\sigma_1 = \sigma_2$, there exists no thresholding admission rules such that IGM and equal opportunity hold at the same time.
\end{lemma}

\begin{proof}
Fix any hiring cost $C$. For simplicity, in the rest of the proof, we let $g_i^* = g_i^*(C)$. A student with grade $g$ in population $i$ is accepted if and only if $g \geq g_i^*(C)$. Therefore, a student with type $t$ in population $i$ gets accepted w.p.
\begin{align*}
x_i(t) \Pr \left[G_i \geq g_i^* | T_i = t \right]
&= x_i(t) \int_{g \geq g_i^*} \phi \left(\frac{g-t}{\gamma} \right) dg
\\& = x_i(t) \left(1 - \Phi \left( \frac{g_i^*-t}{\gamma}\right) \right)
\end{align*}
and we require
\[
\frac{x_1(t)}{x_2(t)} =  \frac{1 - \Phi \left( \frac{g_2^*-t}{\gamma}\right)}{1 - \Phi \left( \frac{g_1^*-t}{\gamma}\right)}
\]
In the case of a thresholding admission rule with threshold $\beta_i$ in population $i$, we have that
\[
x_i(t) = \int_{s \geq \beta_i} \phi(s-t) ds = 1 - \Phi \left( \beta_i - t \right)
\]
and the above equation becomes
\begin{align}\label{eq: cond_equalodds}
\frac{1 - \Phi \left( \beta_1 - t \right)}{1 - \Phi \left( \beta_2 - t \right)} =  \frac{1 - \Phi \left( \frac{g_2^*-t}{\gamma}\right)}{1 - \Phi \left( \frac{g_1^*-t}{\gamma}\right)}
\end{align}
Note that if IGM holds, i.e. if $g_1^* = g_2^* = g^*$, then it must be the case that $\beta_1 = \beta_2 = \beta$. Remember that at a fixed variance $\sigma$, we have that $\E [T_i | G_i =g, S_i \geq \beta ]$ is a strictly increasing function of $\mu_i$ at fixed $\beta, \sigma, g^*$ by Corollary~\ref{cor: increasing_limits}, therefore implying that
\[
\E [T_1 |  G_1=g, S_1 \geq \beta_1 ] > \E [T_2 | G_2 = g, S_2 \geq \beta_2]
\]
for all $g$. It must then be the case that $g_1^* > g_2^*$, which is a contradiction.
\end{proof}

\begin{lemma}\label{lem: equalodd_partial_2}
Fix any school hiring cost $C$. For $\gamma \neq 1$, equal opportunity is impossible with any thresholding admission rule for any two Gaussians with different means $\mu_1 > \mu_2$ and same variances $\sigma_1 = \sigma_2$. When $\gamma = 1$, equal opportunity holds if and only if the thresholds $\beta_1$ and $\beta_2$ can be chosen such that simultaneously,
\[
\beta_1 = g^*_2(C) \text{ and } \beta_2 = g^*_1(C),
\]
or equivalently,
\[
\E \left[ T_ 1 | G_1 = \beta_2, S_1 \geq \beta_1 \right] = \E \left[ T_ 2 | G_2 = \beta_1, S_2 \geq \beta_2 \right] = C.
\]
\end{lemma}

\begin{proof}
Since IGM cannot hold when equal opportunity holds, we can assume that $g_1^* \neq g_2^*$ w.l.o.g. Pick the population numbers such that $g_1^* > g_2^*$. It is then the case that
\[
\frac{1 - \Phi \left( \frac{g_2^*-t}{\gamma}\right)}{1 - \Phi \left( \frac{g_1^*-t}{\gamma}\right)} < 1
\]
which in turns implies that
\[
\frac{1 - \Phi \left( \beta_1 - t \right)}{1 - \Phi \left( \beta_2 - t \right)} < 1
\]
so we need $\beta_1 < \beta_2$. Since for all t,
\[
\frac{1 - \Phi \left( \beta_1 - t \right)}{1 - \Phi \left( \beta_2 - t \right)} =  \frac{1 - \Phi \left( \frac{g_2^*-t}{\gamma}\right)}{1 - \Phi \left( \frac{g_1^*-t}{\gamma}\right)},
\]
the derivatives of the left-hand side and of the right-hand side need be equal for all $t$. Letting $H(t) = \frac{\phi(t)}{1-\Phi(t)}$ the hazard rate of a standard normal random variable, the equality of derivatives can be written
\begin{align*}
&\frac{1 - \Phi \left( \beta_1 - t \right)}{1 - \Phi \left( \beta_2 - t \right)} \cdot \left( H(\beta_1-t) - H(\beta_2-t) \right)
\\&= \frac{1}{\gamma^2} \frac{1 - \Phi \left( \frac{g_2^*-t}{\gamma}\right)}{1 - \Phi \left( \frac{g_1^*-t}{\gamma}\right)} \cdot \left( H\left(\frac{g_1^*-t}{\gamma}\right) - H\left(\frac{g_2^*-t}{\gamma}\right) \right)
\end{align*}
which further simplifies thanks to Equation~\eqref{eq: cond_equalodds} into
\[
 \frac{H(\beta_1-t) - H(\beta_2-t)}{H\left(\frac{g_2^*-t}{\gamma}\right) - H\left(\frac{g_1^*-t}{\gamma}\right)}
= \frac{1}{\gamma^2}
\]
As $t \to +\infty$, $1-\Phi(\beta_i-t), 1- \Phi\left(\frac{g_1^*-t}{\gamma}\right) \to 1$, and as $\beta_1 < \beta_2$, $g_1^* > g_2^*$, $\phi(\beta_1-t) - \phi(\beta_2-t) \sim_{t \to +\infty} \phi(\beta_1-t)$, $\phi\left(\frac{g_2^*-t}{\gamma}\right) - \phi\left(\frac{g_1^*-t}{\gamma}\right) \sim_{t \to +\infty} \phi\left(\frac{g_2^*-t}{\gamma}\right)$
\[
 \frac{H(\beta_1-t) - H(\beta_2-t)}{H\left(\frac{g_2^*-t}{\gamma}\right) - H\left(\frac{g_1^*-t}{\gamma}\right)} \sim_{t \to +\infty} \frac{ \phi(\beta_1-t)}{\phi\left(\frac{g_2^*-t}{\gamma}\right)}
\]
This limit is either $0$ or $+\infty$ unless $\gamma = 1,~\beta_1 = g_2^*$, in which case it is $1$. It must therefore be the case that $\gamma = 1,~\beta_1 = g_2^*$; note that this directly implies that $\beta_2 = g_1^*$ must hold too, otherwise
\[
\frac{1 - \Phi \left( \beta_1 - t \right)}{1 - \Phi \left( \frac{g_2^*-t}{\gamma}\right)} = 1 \neq \frac{1 - \Phi \left( \beta_2 - t \right)}{1 - \Phi \left( \frac{g_1^*-t}{\gamma}\right)}
\]
which contradicts Equation~\eqref{eq: cond_equalodds}.
\end{proof}

\section{Extension to non-deterministic admission rules}

The proofs only need make use of the fact that $g^*(.)$ is a strictly increasing, continuous and differentiable function on domain $\reals$, and are otherwise identical to the proofs for thresholding admission rules. We show that said properties of $g^*(.)$ hold for non-deterministic allocation rules below, in Lemma \ref{lem: increasing_cond_mean} Lemma~\ref{lem: expected_limit}.

\begin{claim}\label{clm: gaussian_product}
Let $\mu_1, \mu_2 \in \reals$, $\sigma_1, \sigma_2 > 0$, and $\phi$ the density of a standard Normal random variable. Then
\[
\frac{\phi \left(\frac{\mu_1 - t}{\sigma_1} \right) \phi \left(\frac{\mu_2 - t}{\sigma_2} \right)}{\int_t \left(\frac{\mu_1 - t}{\sigma_1} \right) \phi \left(\frac{\mu_2 - t}{\sigma_2} \right) dt}
\]
is the probability density function of a Normal random variable with mean
\[
\mu = \frac{\sigma_2^2 \mu_1 + \sigma_1^2 \mu_2}{\sigma_1^2 + \sigma_2^2}
\]
and variance
\[
\sigma^2 = \frac{1}{\frac{1}{\sigma_1^2} + \frac{1}{\sigma_2^2}} = \frac{\sigma_1^2 \sigma_2^2}{\sigma_1^2 + \sigma_2^2}
\]
\end{claim}

\begin{proof}
\begin{align*}
&\phi \left(\frac{\mu_1 - t}{\sigma_1} \right) \phi \left(\frac{\mu_2 - t}{\sigma_2} \right)
\\&= D \exp \left(-\frac{(\mu_1 - t)^2}{2\sigma_1^2} +  \frac{(\mu_2 - t)^2}{2\sigma_2^2}  \right)
\\& = D \exp \left(- \left( \frac{1}{2 \sigma_1^2} + \frac{1}{2 \sigma_2^2}  \right) t^2
+ \left( \frac{\mu_1}{\sigma_1^2} + \frac{\mu_2}{\sigma_2^2}  \right) t
-\left( \frac{\mu_1^2}{2 \sigma_1^2} + \frac{\mu_2}{2 \sigma_2^2} \right)
\right)
\\&= K \exp \left( - \frac{1}{2} \left( \frac{1}{\sigma_1^2} + \frac{1}{\sigma_2^2}\right) \left(t- \frac{\frac{\mu_1}{\sigma_1^2} + \frac{\mu_2}{\sigma_2^2}}{\frac{1}{\sigma_1^2} + \frac{1}{\sigma_2^2}} \right)^2\right)
\end{align*}
for some non-zero constants $D$ and $K$. This concludes the proof.
\end{proof}

\begin{lemma}\label{lem: increasing_cond_mean}
Let $e_i(\mu_i,g) = \E_i \left[T_i | G_i=g, A_i = 1 \right]$ in population $i$. Then i) $e_i(.)$ is strictly increasing in both $g$ and $\mu_i$, ii) $L^- =\lim_{g \to -\infty} e_i(\mu_i, g)$ and $L^+ = \lim_{g \to +\infty} e_i(\mu_i, g)$ exist (and are possibly infinite) and iii) $g \to e_i(\mu_i,g)$ is invertible and its inverse $g^*_i(.)$ is strictly increasing and differentiable on $(L^-,L^+)$.
\end{lemma}

\begin{proof}
For all $g \in \reals$, we have by Lemma~\ref{lem: conditional_derivatives} that
\begin{align*}
\frac{\partial e_i}{\partial g}(g)
&= \E_i \left[T_i^2 | G_i=g, A_i = 1 \right] - \E_i \left[T_i | G_i=g, A_i = 1 \right]^2
\end{align*}
This is exactly the variance of a random variable $T_i(g)$ with probability density function $t \to \Pr \left[T_i = t | G_i=g, A_i = 1 \right]$. Because for all $g$, $x_i(t) P_i(t) \phi\left(\frac{g-t}{\gamma}\right)$ is continuous by Claim \ref{clm: cont}, $\Pr \left[T_i = t | G=g, A_i = 1 \right]$ cannot put all of its probability mass on a single value of $t$, and it must therefore be the case that $T_i(g)$ has positive variance and $\frac{\partial e_i}{\partial g}(g)  > 0$. Therefore $e_i(.)$ must be strictly increasing in $g$; the exact same reasoning as in the proof of Lemma~\ref{lem: conditional_derivatives} can be applied to $\mu_i$ by symmetry, and this shows monotonicity. This also implies that the limits at $\pm\infty$ must exist (but may not be finite).

Finally, $e_i$ is a (differentiable hence) continuous and strictly increasing function so it is invertible; as $e_i$ is differentiable, it directly implies that its inverse is also strictly increasing and differentiable on its domain.
\end{proof}

\begin{lemma}\label{lem: expected_limit}
$L^- = -\infty$ and $L^+ = +\infty$, and $g^*_i(.)$ is strictly increasing on $\reals$.
\end{lemma}

\begin{proof}
Since
\[
x_i(t) = \int_s A_i(s) \phi(s-t) dt
\]
we can write
\[
e_i(g) = \frac{\int_t \int_s t A_i(s) \phi(s-t) \phi\left(\frac{g-t}{\gamma}\right) \phi \left(\frac{\mu_i-t}{\sigma_i} \right) ds~dt}{\int_t \int_s A_i(s) \phi(s-t) \phi\left(\frac{g-t}{\gamma}\right) \phi \left(\frac{\mu_i-t}{\sigma_i} \right) ds~dt}
\]
Using Claim~\ref{clm: gaussian_product} and letting
\[
\mu_i(g) = \frac{\gamma^2 \mu_i + \sigma_i^2 g}{\sigma_i^2 + \gamma^2}
\]
which is an increasing function of $g$ whose limits are $\pm \infty$ when $g \to \pm \infty$, and
\[
\lambda_i = \frac{\gamma^2\sigma_i^2}{\gamma^2 + \sigma_i^2},
\]
we have that
\[
 \phi\left(\frac{g-t}{\gamma}\right) \phi \left(\frac{\mu_i-t}{\sigma_i} \right) = \phi \left(\frac{\mu_i(g)-t}{\lambda_i}\right) \int_t \phi\left(\frac{g-t}{\gamma}\right) \phi \left(\frac{\mu_i-t}{\sigma_i} \right) dt
\]
and
\[
e_i(g) = \frac{\int_t \int_s t A_i(s) \phi(s-t) \phi\left(\frac{\mu(g)-t}{\tilde{\lambda_i}}\right)  dt~ds}{\int_t \int_s A_i(s) \phi(s-t) \phi\left(\frac{\mu(g)-t}{\lambda_i}  \right) dt~ds}
\]

Note that
\[
\left\vert t A_i(s) \phi(s-t) \phi\left(\frac{\mu_i(g)-t}{\lambda_i}\right) \right\vert \leq |t| \phi(s-t) \phi\left(\frac{\mu_i(g)-t}{\lambda_i}\right)
\]
and
\begin{align*}
&\int_t \int_s |t| \phi(s-t) \phi\left(\frac{\mu_i(g)-t}{\lambda_i}\right) ds~dt
\\&= \int_t |t| \phi\left(\frac{\mu_i(g)-t}{\lambda_i}\right) \int_s \phi(s-t) ds~dt
\\&= \int_t |t| \phi\left(\frac{\mu_i(g)-t}{\lambda_i}\right) dt
\end{align*}
exists and is finite for any $g \in \reals$ (as the normal distribution admits a moment generating function), and the same argument holds for
\begin{align*}
\int_t \int_s \phi(s-t) \phi\left(\frac{\mu_i(g)-t}{\lambda_i}\right) ds~dt.
\end{align*}
Hence, Fubini's theorem applies, and we can write
\[
e_i(g) = \frac{\int_s A_i(s) \int_t t \phi(s-t) \phi\left(\frac{\mu_i(g)-t}{\lambda_i}\right) dt~ds}{\int_s A_i(s) \int_t \phi(s-t) \phi\left(\frac{\mu_i(g)-t}{\lambda_i}\right) dt~ds}
\]
By Claim~\ref{clm: gaussian_product},
\begin{align*}
&\int_t t \phi(s-t) \phi\left(\frac{\mu_i(g)-t}{\lambda_i}\right) dt
\\& = \left(\frac{\lambda_i^2 s + \mu_i(g)}{\lambda_i^2+1} \right) \cdot \int_t \phi(s-t) \phi\left(\frac{\mu_i(g)-t}{\lambda_i}\right) dt
\end{align*}
and the expectation can be written
\[
e_i(g) = \frac{\mu_i(g)}{\lambda_i^2+1} + \frac{\lambda_i^2}{\lambda_i^2+1} \frac{\int_s s A_i(s) \int_t \phi(s-t) \phi\left(\frac{\mu_i(g)-t}{\lambda_i}\right) dt~ds}{\int_s A_i(s) \int_t \phi(s-t) \phi\left(\frac{\mu_i(g)-t}{\lambda_i}\right) dt~ds}
\]
We further have that,
\begin{align*}
\phi(s-t) \phi\left(\frac{x-t}{a}\right) = \frac{1}{2\pi a} \exp \left(-(x-t)^2/2-(x-t)^2/2a^2 \right)
 \end{align*}
and we can write
\begin{align*}
&(s-t)^2+(x-t)^2/a^2
\\&= \frac{a^2+1}{a^2} t^2 - 2 \left( \frac{a^2 s+x}{a^2+1}\right) \frac{a^2+1}{a^2} t + \frac{x^2+a^2 s}{a^2}
\\&= \frac{a^2+1}{a^2} \left(t- \frac{a^2 s+x}{a^2+1} \right)^2 + \frac{x^2+a^2s}{a^2} - \frac{(a^2 s + x)^2}{a^2 (a^2+1)}
\\& = \frac{a^2+1}{a^2} \left(t- \frac{a^2 s+x}{a^2+1} \right)^2 + \frac{(x-s)^2}{a^2+1}
\end{align*}
Therefore,
\begin{align*}
&\int_t \phi(s-t) \phi\left(\frac{\mu_i(x)-t}{a_i}\right) dt
\\&= \frac{1}{2\pi a} \exp\left(-\frac{(x-s)^2}{2(a^2+1)} \right) \int_t \exp\left(-\frac{a^2+1}{2a^2} \left(t- \frac{a^2 s+x}{a^2+1} \right)^2 \right) dt
\\&= \frac{1}{2\pi a} \exp\left(-\frac{(x-s)^2}{2(a^2+1)} \right) \cdot \sqrt{2\pi} \frac{a}{\sqrt{a^2+1}}
\\&= \frac{1}{\sqrt{2 \pi (a^2+1)}} \exp\left(-\frac{(x-s)^2}{2(a^2+1)} \right)
\\&= \phi\left(\frac{x-s}{\sqrt{a^2+1}} \right)
\end{align*}
and
\[
e_i(g) = \frac{\mu_i(g)}{\lambda_i^2+1} + \frac{\lambda_i^2}{1+\lambda_i^2} \frac{\int_s s A_i(s) \phi\left(\frac{\mu_i(g)-s}{\sqrt{\lambda_i^2+1}} \right) ds}{\int_s A_i(s) \phi\left(\frac{\mu_i(g)-s}{\sqrt{\lambda_i^2+1}} \right) ds}
\]
Note that
\[
\int_s A_i(s) \phi\left(\frac{\mu_i(g)-s}{\sqrt{\lambda_i^2+1}} \right) ds > 0
\]
for monotone and non-zero $A_i(s)$, and
\[
A_i(s) \phi\left(\frac{\mu_i(g)-s}{\sqrt{\lambda_i^2+1}} \right),~ s A_i(s) \phi\left(\frac{\mu_i(g)-s}{\sqrt{\lambda_i^2+1}} \right)
\]
are absolutely integrable as $|A_i(s)|\leq 1$ hence integrable for all $g$, so
\[
\frac{\int_s s A_i(s) \phi\left(\frac{\mu_i(g)-s}{\sqrt{\lambda_i^2+1}} \right) ds}{\int_s A_i(s) \phi\left(\frac{\mu_i(g)-s}{\sqrt{\lambda_i^2+1}} \right) ds}
\]
is finite for all $g$, and in particular for $g=0$. By the same argument as in the proof of Lemma~\ref{lem: conditional_derivatives}, $e_i(g)$ is a non-decreasing function of $\mu_i(g)$ and hence of $g$. Therefore, for $g < 0$ we have that
\[
e_i(g) \leq \frac{\mu_i(g)}{\lambda_i^2+1} + \frac{\lambda_i^2}{1+\lambda_i^2} \frac{\int_s A_i(s) \phi\left(\frac{\mu_i(0)-s}{\sqrt{\lambda_i^2+1}} \right) ds}{\int_s A_i(s) \phi\left(\frac{\mu_i(0)-s}{\sqrt{\lambda_i^2+1}} \right) ds}
\]
which implies $L^- = -\infty$, and similarly $L^+ = +\infty$ as for $g > 0$,
\[
e_i(g) \geq \frac{\mu_i(g)}{\lambda_i^2+1} + \frac{\lambda_i^2}{1+\lambda_i^2} \frac{\int_s A_i(s) \phi\left(\frac{\mu_i(0)-s}{\sqrt{\lambda_i^2+1}} \right) ds}{\int_s A_i(s) \phi\left(\frac{\mu_i(0)-s}{\sqrt{\lambda_i^2+1}} \right) ds}
\]
\end{proof}

\end{document}